\documentclass[lettersize,journal]{IEEEtran}
\usepackage{amsmath,amsfonts}
\usepackage[thmmarks,amsmath]{ntheorem}
\usepackage{algorithmic}
\usepackage{algorithm}
\usepackage{array}
\usepackage[caption=false,font=normalsize,labelfont=sf,textfont=sf]{subfig}
\usepackage{textcomp}
\usepackage{stfloats}
\usepackage{url}
\usepackage{verbatim}
\usepackage{graphicx}
\usepackage{cite}
\usepackage{amssymb}
\usepackage{physics}
\usepackage{epstopdf}
\usepackage{cases}
\usepackage{color}
\usepackage{indentfirst}
\usepackage{amsmath}
\allowdisplaybreaks[4]

\newtheorem*{proof}{Proof:}
\theoremseparator{:}

\newtheorem{theorem}{Theorem}
\newtheorem{lemma}{Lemma}

\begin{document}

	\title{Structured Connectivity for 6G Reflex Arc: Task-Oriented Virtual User and New Uplink-Downlink Tradeoff}

	\author{Xinran~Fang, Chengleyang Lei,
	Wei~Feng,~\IEEEmembership{Senior Member,~IEEE,} Yunfei~Chen,~\IEEEmembership{Senior Member,~IEEE,} Ning~Ge, and Shi~Jin,~\IEEEmembership{Fellow,~IEEE}
	\thanks{X.~Fang, C. Lei, W.~Feng, and N. Ge are with the State Key Laboratory of Space Network and Communications, Department of Electronic Engineering, Tsinghua University, Beijing 100084, China (e-mail: {fxr20}@mails.tsinghua.edu.cn, {lcly21}@mails.tsinghua.edu.cn, {fengwei}@tsinghua.edu.cn, and  {gening}@tsinghua.edu.cn.)
	Y. Chen is with the Department of Engineering, University of Durham, DH1 3LE Durham, U.K. (e-mail: yunfei.chen@durham.ac.uk). S. Jin is with the National Mobile Communications Research Laboratory, Southeast University, Nanjing 210096, China (e-mail: jinshi@seu.edu.cn).}\vspace{-1.6em}}

	\maketitle

	\begin{abstract}

		To accommodate the evolving demands of unmanned operations, the future sixth-generation (6G) network will support not only communication links but also sensing-communication-computing-control ($\mathbf{SC}^3$) loops. In each $\mathbf{SC}^3$ cycle, the sensor uploads sensing data to the computing center, and the computing center calculates the control command and sends it to the actuator to take action. To maintain the task-level connections between the sensor-computing center link and the computing center-actuator link, we propose to treat the sensor and actuator as a virtual user. In this way, the two communication links of the $\mathbf{SC}^3$ loop become the uplink and downlink (UL\&DL) of the virtual user. Based on the virtual user, we propose a task-oriented UL\&DL optimization scheme. This scheme jointly optimizes UL\&DL transmit power, time, bandwidth, and CPU frequency to minimize the control linear quadratic regulator (LQR) cost.  We decouple the complex problem into a convex UL\&DL bandwidth allocation problem with the closed-form solution for the optimal time allocation. Simulation results demonstrate that the proposed scheme achieves a task-level balance between the UL\&DL, surpassing conventional communication schemes that optimize each link separately.

	\end{abstract}

\vspace{-0.5em}
	\section{Introduction}
	Currently, wireless communication has expanded from connecting humans to connecting machines. Driven by the increasing demands for unmanned operations, supporting field robots has been identified as an important use case for the sixth-generation (6G) network \cite{ref0}. Take disaster rescue as an example. Once the emergency happens, the network swiftly connects sensors, actuators, and the computing center into sensing-communication-computing-control ($\mathbf{SC}^3$) loops, which execute various tasks based on the periodical control. In each $\mathbf{SC}^3$ cycle, the sensor uploads the collected data to the computing center, the computing center calculates the control command and sends it to the actuator, and the actuator takes action. Relying on the effective feedback, the $\mathbf{SC}^3$ loops dynamically adapts to varying environments and finishes different tasks without human intervention.

	As shown in Fig. \ref{fig1}, the $\mathbf{SC}^3$ loop bears a striking resemblance to the reflex arc. Dating back to the 16th century, when Descartes conceptualized the reflex arc, biological researchers found that \emph{``the presence of reflexes is only dependent on the functional integrity of the components of the reflex arc”} \cite{bio1}. Similarly, the task-execution capability of the $\mathbf{SC}^3$ loop is only dependent on the integrity of its components. 
   	However, the fifth-generation (5G) network mainly serves humans and takes the communication link as the basic unit \cite{feng}. It applies the multiple access technique to separate different users and applies the duplex technique to separate the uplink and downlink (UL\&DL). As a result, the $\mathbf{SC}^3$ loop is considered as two independent links: one link from the sensor to the computing center and another link from the computing center to the actuator. Although this link-level decomposition brings high capacity to the network,  it breaks the task-level connections among the $\mathbf{SC}^3$ loop components, making it less effective for the task-oriented networks in 6G.

    To investigate the $\mathbf{SC}^3$ loop, the founder of cybernetics, Norbert Wiener, emphasized that \emph{``the problems of control engineering and of communication engineering were inseparable, and that they centered not around the technique of electrical engineering but around the much more fundamental notion of the message"} \cite{winner}. Therefore, when we are concerned about the \emph{``fundamental notion of the message"}—the information usage behind data transmission, it is easy to find that the two communication links within one $\mathbf{SC}^3$ loop are not independent for their own transmission but are interconnected for executing a common task.  To maintain the task-level connections between the two links, we propose to regard the sensor and actuator as a virtual user. Consequently, the link between the sensor and the computing center, along with the link between the computing center and the actuator, are considered the UL\&DL of the virtual user. This virtual-user approach allows us to take the UL\&DL as two interconnected links and jointly configure their communication capabilities. Thus, the UL\&DL could be aligned to support the functioning of the $\mathbf{SC}^3$ loop.

	In the literature, few works investigated the joint UL\&DL optimization given that frequency division duplexing (FDD) and time division duplexing (TDD) are commonly used to separate UL\&DL. In the early stage, El-Hajj \emph{et al.} addressed the demands for balanced UL\&DL and proposed a sum-rate maximization scheme. The UL\&DL rate difference was considered as a constraint in the optimization \cite{ref1}. Then, the emerging time-sensitive applications put high requirements on the round-time trip latency, leading to the joint UL\&DL design \cite{ref2,ref3, ref4, ref5}. The work in \cite{ref2,ref3, ref4} investigated the multi-access edge computing system, and proposed the energy/power minimization schemes under the round-time trip latency constraints.
    The work in \cite{ref5} proposed a bandwidth minimization scheme that optimizes the UL\&DL bandwidth and delay components.
    However, these studies still focused on communication-related metrics such as spectrum efficiency and transmission latency, while ignored the task performance behind the data transmission. When the focus shifts to the $\mathbf{SC}^3$ loop, in the field of wireless control system (WCS),  researchers investigated the similar control closed loop, and they paid attention to the control performance behind communication \cite{ref6}.
    For example, Gatsis \emph{et al.} focused on the uplink (UL) transmission and proposed a power adaptation scheme, which adapts the sensor power to both plant and channel states to minimize the linear quadratic regulator (LQR) cost \cite{Gatsis}. Extending to the multi-$\mathbf{SC}^3$ loops, Chang \emph{et al.} optimized UL bandwidth, power, and control convergence rate to maximize the spectrum efficiency \cite{Chang}. Focusing on the downlink (DL) transmission, Fang \emph{et al.} optimized the transmit power and block length \cite{Fang}, and Ana \emph{et al.} optimized the bandwidth \cite{Ana}. These studies provided great insights for the $\mathbf{SC}^3$-loop design. However, they all focused on one link and assumed the other link was ideal. There lacks the work that takes the $\mathbf{SC}^3$ loop as an integrated structure and jointly configure UL\&DL from a task-oriented perspective.

    In this letter, we investigate the $\mathbf{SC}^3$ loop and regard the sensor and actuator as a virtual user. Based on the virtual user, we propose a task-oriented joint UL\&DL optimization scheme. Using the control LQR cost as the objective, we jointly optimize UL\&DL transmit power, time, bandwidth, and computing CPU frequency. This complex problem is then simplified into a convex bandwidth allocation problem with the closed-form solution for the time allocation. The optimal LQR cost is also expressed in a closed-form. Our simulation results underscore the superiority of our task-oriented UL\&DL optimization scheme, showing that it achieves a task-level balance between UL\&DL.

	\section{Virtual-User-Based $\mathbf{SC}^3$ Loop Model}
	\begin{figure*} [t]
		\centering
		\includegraphics[width=1\linewidth]{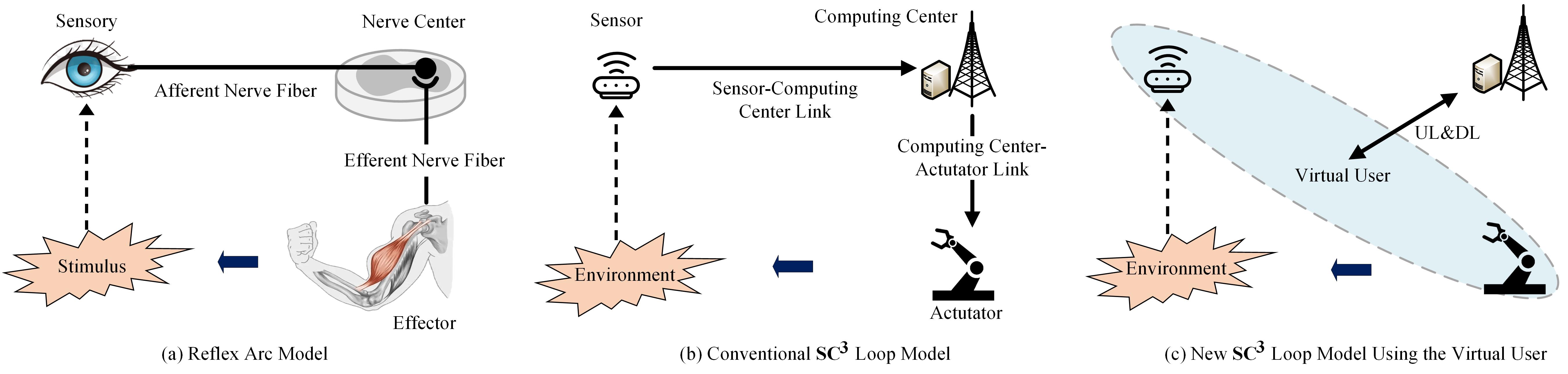}
		\caption{Illustration of the reflex arc model, conventional $\mathbf{SC}^3$ loop model, and new $\mathbf{SC}^3$ loop model using the virtual user.}
		\label{fig1}
	\end{figure*}
	As shown in Fig. \ref{fig1}, the $\mathbf{SC}^3$ loop encompasses three nodes: the sensor, the computing center, and the actuator, connected by two communication links: one between the sensor and the computing center, and another between the computing center and the actuator. To maintain the task-level connections among $\mathbf{SC}^3$ loop components, we regard the sensor and actuator as a single virtual user. In this setup, the sensor utilizes the UL to upload sensing data, while the actuator employs the DL to receive commands. The considered $\mathbf{SC}^3$ loop is executing a control-type task. Without loss of generality, we model the controlled system as a linear time-invariant system. The discrete-time evolution equation is given by,
	\begin{equation}
		\label{sta_elv}
		\mathbf{x}_{i+1} = \mathbf{A}\mathbf{x}_{i}+\mathbf{B}\mathbf{u}_{i}+\mathbf{v}_{i},
	\end{equation}
    where $i$ is the time index, $\mathbf{x}_{i}\in\mathbb{R}^{n\times1}$ is the system state, $\mathbf{u}_{i}\in\mathbb{R}^{m\times1}$ is the control action, $\mathbf{v}_{i}\in\mathbb{R}^{n\times1}$ is the process noise, and $\mathbf{A}\in\mathbb{R}^{n\times n}$ and $\mathbf{B}\in\mathbb{R}^{n\times m}$ are determined by the system dynamics. To measure the control performance, the LQR cost is used, which is a weighted summation of the state derivation and control input,
	\begin{equation}
		l=	\limsup\limits_{N\rightarrow \infty}\mathbb{E} \left[ \sum_{i=1}^{N} \left(\mathbf{x}_{i}^\text{T}\mathbf{Q}\mathbf{x}_{i} +\mathbf{u}_{i}^\text{T}\mathbf{R}\mathbf{u}_{i}\right) \right],
	\end{equation}
	where $\mathbf{Q}\in\mathbb{R}^{n\times n}$ and $\mathbf{R}\in\mathbb{R}^{m\times m}$ are two weight matrices.

	Denote UL\&DL transmit power, time, and bandwidth as $p_{u/d}$, $t_{u/d}$, and $B_{u/d}$, respectively. For simplicity, we use the subscript $u/d$  to represent UL/DL.
    Within an $\mathbf{SC}^3$ cycle, the UL\&DL transmit information, denoted as $D_{u/d} \ $, is given by:
		\begin{equation}
		D_{u/d}\leqslant t_{u/d}R(p_{u/d},B_{u/d}) \quad (\text{bits}/\mathbf{SC}^3 \text{cycle}),
	\end{equation}
 where $R(p_{u/d},B_{u/d})$ denotes the UL\&DL rate,
 \begin{equation}
 	\label{rate}
 	R(p_{u/d},B_{u/d})\triangleq B_{u/d}\log_2(1+\frac{p_{u/d}|h_{u/d}|^2}{B_{u/d}N_0}) \quad (\text{bits/s}),
 \end{equation}
 where $h_{u/d}$ denotes the channel gain, which is assumed to be constant along the control process, and $N_0$ denotes the noise power spectral density. For computing, it is
 modeled as an information-extraction process, which can be described as,
	\begin{equation}
		D_u\rightarrow \rho D_u,
	\end{equation}
    where $\rho\ (0<\rho<1)$ denotes the information extraction ratio, and $\rho D_u$ quantifies the task-related information extracted from the sensing data. It also quantifies the information contained in the command.
    Denote $f$ as the CPU frequency and the computing time is calculated by,
	\begin{equation}
		t_c=\frac{\alpha D_u}{f},
	\end{equation}
    where $\alpha \quad (\text{cycle/bit})$  denotes the required CPU cycles to process one-bit raw data, which quantifies the computation complexity.
    Then, the command is sent to the actuator. Considering the DL capacity constraint, the task-related information received by the actuator is jointly determined by UL\&DL:
	\begin{equation}
		\label{f7}
		D_{\mathbf{SC}^3}=\min (\rho D_u, D_d) \quad (\text{bits}/\mathbf{SC}^3 \text{cycle}).
	\end{equation}
   $D_{\mathbf{SC}^3}$ also denotes the information that finally takes effect for the control task within an $\mathbf{SC}^3$ cycle, which we refer to as  the closed-loop entropy rate. Denoting the time for accomplishing one $\mathbf{SC}^3 \ \text{cycle}$ as $T$, we have the following cycle-time constraint for UL\&DL and computing time:
   \begin{equation}
   	t_u+t_c+t_d\leqslant T.
   \end{equation}
   Moving forward, the lower bound of the LQR cost has the following relationship with the closed-loop entropy rate \cite{ref14}:
    \begin{equation}\label{f6}
    l \geqslant \frac{n N \!\left( \mathbf{v}\right)|\det \mathbf{M}|^\frac{1}{n}} {2^{\frac{2}{n}(D_{\mathbf{SC}^3}-\log_2|\det \mathbf{A}|)}-1}+\text{tr}\left( \mathbf{\Sigma}_{\mathbf{v}}\mathbf{S}\right),
    \end{equation}
    where $N(\mathbf{x})\triangleq\frac{1}{2\pi e}e^{\frac{2}{n}h(\mathbf{x})}$, $h(\mathbf{x})\triangleq-\int_{\mathbb{R}^n}f_{\mathbf{x}}(x)\log f_{\mathbf{x}}(x)dx$, $\mathbf{\Sigma}_{\mathbf{v}}$ is the covariance matrix of the process noise, and $\log_2|\det\mathbf{A}|$ is the intrinsic entropy rate. To ensure the system can be stabilized, the closed-loop entropy rate  needs to exceed the intrinsic entropy rate,
    	\begin{equation}
    		\label{stable}
    		D_{\mathbf{SC}^3}>\log_2|\det \mathbf{A}|.
    	\end{equation}
    The above stability condition is both necessary and sufficient. Provided \eqref{stable} is satisfied, there must exist a code-controller that ensures the asymptotic stability, i.e., $\lim\limits_{i\rightarrow +\infty} \sup\parallel\mathbf{x}_i\parallel<\infty$ \cite{Nair}. Therefore, \eqref{f6} makes sense in the stable region, otherwise, the LQR cost is infinite. $\mathbf{M}$ and $\mathbf{S}$ are solved by the following Riccati equations:
	\begin{equation}\label{Riccati}
			\mathbf{S} = \mathbf{Q} + \mathbf{A}^\text{T}\left(\mathbf{S}- \mathbf{M}\right) \mathbf{A}, \
			\mathbf{M}  = \mathbf{S}^T \mathbf{B} \left( \mathbf{R} + \mathbf{B} \mathbf{S} \mathbf{B}\right)^{-1} \mathbf{B}^\text{T} \mathbf{S}.
	\end{equation}

	\section{Task-Oriented UL\&DL Optimization Scheme}
	In the considered $\mathbf{SC}^3$ loop model, we can see that the LQR cost is bounded by the closed-loop entropy rate, which, in turn, is determined by the UL\&DL capabilities. Additionally, the computing speed influences the UL\&DL time and indirectly influences the control performance.
	Therefore, we jointly optimize the UL\&DL transmit power, time, bandwidth, and CPU frequency to minimize the LQR cost,
	\begin{subequations}
	\begin{align}
		\mbox{(P1)} \ \ &\min\limits_{p_u,t_u,B_u,f,p_d,t_d,B_d,\geqslant 0} l \\
		\text{s.t.} \ &l \geqslant \frac{n N \!\left( \mathbf{v}\right)|\det \mathbf{M}|^\frac{1}{n}} {2^{\frac{2}{n}(D_{\mathbf{SC}^3}-\log_2|\det \mathbf{A}|)}-1}+\text{tr}\left( \mathbf{\Sigma}_{\mathbf{v}}\mathbf{S}\right) \label{a1} \\
		&D_{\mathbf{SC}^3}>\log_2|\det \mathbf{A}|S \label{12c} \\
		&D_{\mathbf{SC}^3}\leqslant \min(\rho D_u,D_d) \label{b4}\\
		&D_{u/d}\leqslant t_{u/d}R(p_{u/d},B_{u/d}) \label{b7}\\
		&t_u+\frac{\alpha D_u}{f}+t_d\leqslant T \label{b6}\\
		&B_u+B_d\leqslant B_{\max},  p_{u/d}\leqslant P_{\text{umax}/\text{dmax}}, f\leqslant f_{\max}, \label{f13}
	\end{align}
	\end{subequations}
 where \eqref{f13} represents the resource constraints on the bandwidth, transmit power, and  CPU frequency, with the maximal value denoted by $B_{\max}$, $P_{\text{umax}/\text{dmax}}$, and $f_{\max}$. Given that the right-hand side of \eqref{a1} is a decreasing function of $D_{\mathbf{SC}^3}$, minimizing the LQR cost is equivalent to maximizing the closed-loop entropy rate. In addition, it is easy to find that the optimal solution necessitates the full utilization of the transmit power and CPU frequency, i.e., $p_{u/d}^*=P_{\text{umax}/\text{dmax}}$, and $f^*=f_{\text{max}}$. On this basis, we denote  $R(B_{u/d})\triangleq R(p_{u/d}^*,B_{u/d})$. Then, (P1) is simplified into an UL\&DL time and bandwidth allocation problem as follows:
	    \begin{subequations}
		\begin{align}
			\mbox{(P2)} \ \ &\max\limits_{t_u,B_u,t_d,B_d\geqslant 0}  D_{\mathbf{SC}^3} \\
			\text{s.t.}  \ \
			&D_{\mathbf{SC}^3}\leqslant \min(\rho D_u,D_d) \\
			&D_{u/d}\leqslant t_{u/d}R(B_{u/d}) \label{a2} \\
			&t_u+\frac{\alpha D_u}{f_{\max}}+t_d\leqslant T \label{f10}\\
			&B_u+B_d\leqslant B,
		\end{align}
	\end{subequations}
	where we omit the stability condition \eqref{12c} and test it after solving (P2). If the optimal closed-loop entropy rate, $(D_{\mathbf{SC}^3})^*$, satisfies the stability condition, the LQR cost is calculated by \eqref{f6}. Otherwise, the system cannot be stabilized and the LQR cost is infinite.
    \begin{lemma}
        The optimal solution to (P2) is to achieve a task-level balance between UL\&DL:
        \begin{equation}
        \rho D_u^*=D_d^*. \label{35}
        \end{equation}
    \end{lemma}
    \begin{proof}
    To maximize the closed-loop entropy rate, \eqref{a2} must be satisfied as an equality at the optimal solution: $D_{u/d}^*=t_{u/d}R(B_{u/d})$. If \eqref{35} does not hold at the optimal solution, e.g., $\rho D_u^*>D_d^*$,  the optimal closed-loop entropy rate is determined by $D_{\mathbf{SC}^3}^*=D_d^*$. However, $D_{\mathbf{SC}^3}^*$ can be further increased by reallocating the UL time to the DL: $t_{d}R(B_{d})\uparrow \ \rightarrow \ D_d^*\uparrow\ \rightarrow \ D_{\mathbf{SC}^3}^* \uparrow$,  as long as  $\rho D_u^*>D_d^*$ holds.
    This contradicts to the assumption that $D_{\mathbf{SC}^3}^*$ is optimal. We can use the similar reasoning to falsify $\rho D_u^* < D_d^*$. Therefore, $\rho D_u^*= D_d^*$ must hold at the optimal solution.\hfill $\blacksquare$\par
    \end{proof}
    By leveraging the task-level balance between UL\&DL, we show that (P2) can be converted into a convex bandwidth allocation problem with the closed-form solution for the time allocation.
    \begin{theorem}
    The optimal UL\&DL bandwidth allocation is achieved by solving the following problem:
     \begin{subequations}
	\begin{align}
		\mbox{(P3)} \ \ \min\limits_{B_u,B_d\geqslant 0} \ \ &\frac{1}{\rho R(B_u)}+\frac{1}{ R(B_d)} \\
		\text{s.t.} \ \
		&B_u+B_d\leqslant B_{\max}.
	\end{align}
	\end{subequations}
    The optimal UL\&DL time are given by,
    \begin{equation}
			\label{f11}
				t_u^*=\frac{\frac{1}{\rho R_u^*} T}{\frac{1}{\rho R_u^*}+\frac{1}{R_d^*}+\frac{ \alpha }{\rho f_{\max}}},
				\quad t_d^*=\frac{\frac{1}{R_d^*}T}{\frac{1}{\rho R_u^*}+\frac{1}{R_d^*}+\frac{\alpha}{\rho f_{\max}}}.\\
    \end{equation}
    where $R_{u/d}^*\triangleq R(p_{u/d}^*, B_{u/d}^*)$ and $B_{u/d}^*$ is the solution to (P3). The optimal closed-loop entropy rate and LQR cost are given by,
	\begin{equation}
		\label{b3o}
		\begin{aligned}
			D_{\mathbf{SC}^3}^*=\frac{T}{\frac{1}{\rho R_u^*}+\frac{1}{R_d^*}+\frac{\alpha}{ \rho f_{\max}}},
		\end{aligned}
	\end{equation}
	\begin{equation}
		\label{f15}
		l^*= \left\{
		\begin{aligned}
			&+\infty,  \quad\quad \quad\quad\quad\quad \quad\quad \ D_{\mathbf{SC}^3}^*\leqslant \log_2|\det \mathbf{A}|,&\\
			&\small \frac{n N \!\left( \mathbf{v}\right)|\det \mathbf{M}|^\frac{1}{n}} {2^{\frac{2}{n}\bigg(D_{\mathbf{SC}^3}^*-\log_2|\det \mathbf{A}|\bigg)}-1}+\text{tr}\left( \mathbf{\Sigma}_{\mathbf{v}}\mathbf{S}\right), \text{otherwise.}&
		\end{aligned}
		\right.
	\end{equation}
    \end{theorem}
    \begin{proof}
   According to \eqref{a2} and \eqref{35}, the optimal DL time can be expressed as a function of the UL time:
    \begin{equation}
    \label{b1}
        \rho t_u^*R(B_u)=t_d^*R(B_d) \Rightarrow t_d^*=\frac{\rho t_u^*R(B_u)}{R(B_d)}.
    \end{equation}
    The computing time can also be expressed as a function of the UL time,
    \begin{equation}
    \label{b2}
        t_c^*=\frac{\alpha D_u^*}{f_{\max}}=\frac{\alpha t_u^*R(B_u)}{f_{\max}}.
    \end{equation}
    By substituting \eqref{b1} and \eqref{b2} into \eqref{f10}, the optimal UL time is expressed as a function of the bandwidth:
    \begin{equation}
    \label{b5}
  t_u^*+t^*_c+t^*_d=T        \Rightarrow
         t_u^*=\frac{\frac{1}{\rho R(B_u)} T}{\frac{1}{\rho R(B_u)}+\frac{1}{R(B_d)}+\frac{\alpha}{\rho f_{\max}}}.
    \end{equation}
    On this basis, the optimal closed-loop entropy rate  can be expressed as the function of the bandwidth:
    \begin{equation}
    \label{b3}
			D_{\mathbf{SC}^3}^*=\rho D_u^*
			=\rho  t_u^*R(B_u)=\frac{T}{\frac{1}{\rho R(B_u)}+\frac{1}{R(B_d)}+\frac{\alpha}{ \rho f_{\max}}}.
	\end{equation}
    From \eqref{b3}, we can find that maximizing the closed-loop entropy rate is to minimize $[\frac{1}{ R(B_u)}+\frac{\rho}{R(B_d)}]$ and thus we get (P4). By solving (P4) and substituting the optimal UL\&DL bandwidth, $B_u^*$ and $B_d^*$, into the \eqref{b3} and \eqref{b5}, we get $t_u^*$ and $D_{\mathbf{SC}^3}^*$. We could further get $t_d^*$ by substituting $t_u^*$ into \eqref{b1} and get $l^*$  by testing the stability condition \eqref{stable} and substituting  $D_{\mathbf{SC}^3}^*$ into \eqref{f6}. \hfill $\blacksquare$\par
    \end{proof}
In fact, (P4) is a time-minimization problem. Its objective, $[\frac{1}{\rho R(B_u)}+\frac{1}{ R(B_d)}]$, represents the  UL\&DL time for transmitting one-bit task-related information. Given that maximizing the closed-loop entropy rate is equivalent to minimize the time for transmitting one-bit task-related information from the sensor to the actuator, the equivalence of (P3) and (P4) becomes obvious. We summarized the proposed in {\bf{Algorithm \ref{Tab1}}}.
\begin{algorithm}[t]
	\caption{The Proposed Optimization Algorithm} \small
	\label{Tab1}
	\begin{algorithmic}[1]
		\REQUIRE
		Control related parameters: $n$, $\log_2|\det \mathbf{A}|$,  $\mathbf{B}$, $\mathbf{Q}$, $\mathbf{R}$, $T$ and $\Sigma_{\mathbf{v}}$;
			Communication related parameters: $P_{\text{umax}}$, $P_{\text{dmax}}$, $B_{\max}$, $h_u$, $h_d$, and $N_0$;
			Computing related parameters: $\alpha$, $\rho$, and $f_{\max}$;
		\STATE  Calculate $\mathbf{S}$ and $\mathbf{M}$ according to \eqref{Riccati};
		\STATE  Calculate the optimal UL\&DL transmit power and CPU frequency: $p^*_{u/d}=p_{\text{umax/dmax}}$ and $f^*=f_{\max}$;
		\STATE  Solve (P3) to obtain the optimal bandwidth allocation, $B_{u/d}^*$;
		\STATE	Calculate $R_{u/d}^*$ according to \eqref{rate}, and calculate the optimal time allocation according to  \eqref{f11};
		\STATE	Calculate the optimal closed-loop entropy rate, {$D_{\mathbf{SC}^3}^*$}, according to \eqref{b3o};
		\STATE Judge the stability condition \eqref{stable} and calculate the limit LQR cost, $l^*$, according to \eqref{f15}.
	\end{algorithmic}
\end{algorithm}

\begin{figure} [t]
	\centering
	\includegraphics[width=0.9\linewidth]{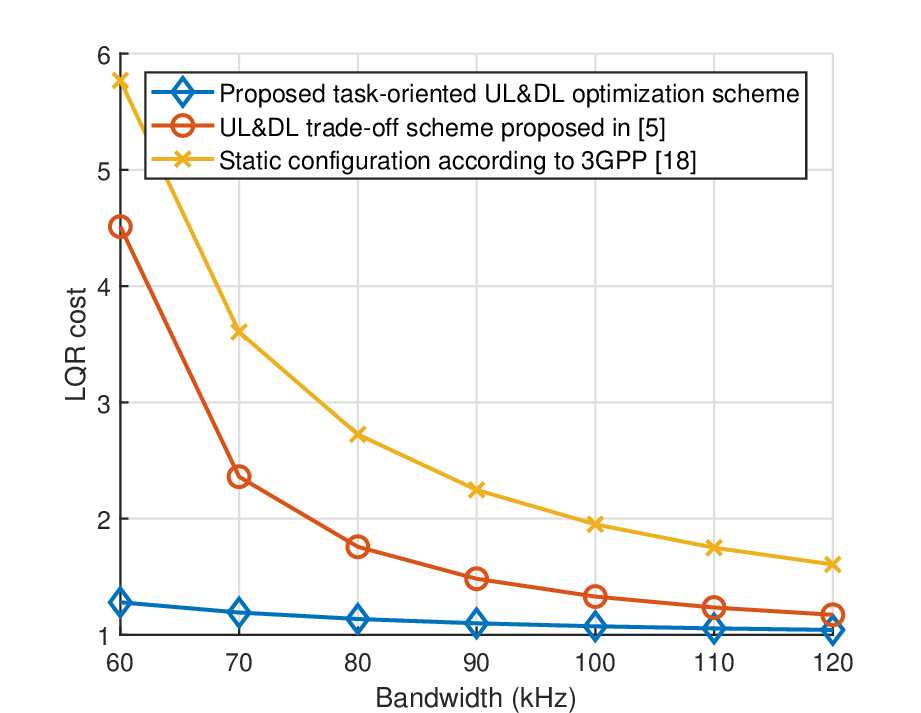}
	\caption{The LQR cost varies with the available bandwidth resources under three UL\&DL configuration schemes. }
	\label{simu1}
\end{figure}

\label{sec_simulation}
In this section, we present simulation results. Simulation parameters are set as: $P_{\text{umax}} = 0.1$ W, $P_{\text{dmax}} = 1$ W, $B_{\text{max}}=1$ MHz, $N_0$=-174 \text{dBm}, and the channel gain are calculated by the path-loss model of $[32.4+20*\log_{10}(f_c)+20*\log_{10}(d_{u/d})] \text{(dB)}$ \cite{refd}, where $f_c$ (MHz) represents the carrier frequency and $d_{u/d}$ (km) represents the transmission distance, i.e., $f_c=2000$ MHz and $d_{u/d}=1$ km. Control related parameters are given by $n=m=100$, $\log_2|\det\mathbf{A}|=50$,  $\mathbf{R}=\mathbf{0}_{100}$, $\mathbf{Q}=\mathbf{I}_{100}$, $\Sigma_{\mathbf{v}}$=$0.01*\mathbf{I}_{100}$  , and $T=20$ ms \cite{Fang}. Computing-related parameters are given by $f_{\text{max}} = 1$ GHz, $\alpha = 100$ cycles/bit, and $\rho=0.01$.

\section{Simulation Results and Discussion}

In Fig. \ref{simu1}, we compare the proposed scheme with the UL\&DL trade-off scheme proposed in \cite{ref1} and the static configuration. For fair comparison, we replicate \cite{ref1} by adopting the sum rate of UL\&DL as the objective, $[\rho D_u+D_d]$, and constraining the UL\&DL disparities using the constraint $|\rho D_u-D_d| \leqslant D_0$, where $D_0=100$ (bits). For the static setup, we apply the equal bandwidth division based on FDD and set $t_u=\frac{6T}{7}, t_c=\frac{T}{14}$, and $t_d=\frac{T}{14}$, which aligns with the slot format 34 in the 5G New Radio (NR) TDD standard \cite{3gpp_r17}.  We can see that the proposed scheme consistently exhibits the lowest LQR cost compared with the other two schemes. This outcome highlights the superiority of jointly configuring UL\&DL from a task-oriented perspective.

We further present Fig. \ref{simu2} to reveal the reason behind the superiority of the proposed scheme. In this simulation, the maximum bandwidth is set as $B_{\max}=1$ MHz. Fig. \ref{simu2} is a dual-axis chart, with the left axis representing the task-related information. From the bars, we can see that under the proposed scheme, the UL\&DL are aligned to transmit the same amount of task-related information within an $\mathbf{SC}^3$ cycle, while the other two schemes fail to achieve this balance. The scheme in \cite{ref1} allocates more resources to the UL, while the static configuration \cite{3gpp_r17} allocates more resources to the DL. This imbalance results in the DL under the scheme in \cite{ref1} and the UL under the static configuration \cite{3gpp_r17}  becoming the bottleneck, limiting the control performance of the $\mathbf{SC}^3$ loop. As a result, the LQR costs  (black curve, measured by the right axis)  under these two schemes are  higher than the proposed scheme, with a greater imbalance corresponding to a higher LQR cost. This verifies the critical importance of maintaining a task-level UL\&DL balance for the $\mathbf{SC}^3$ loop.

\begin{figure} [t]
	\centering
	\includegraphics[width=1\linewidth]{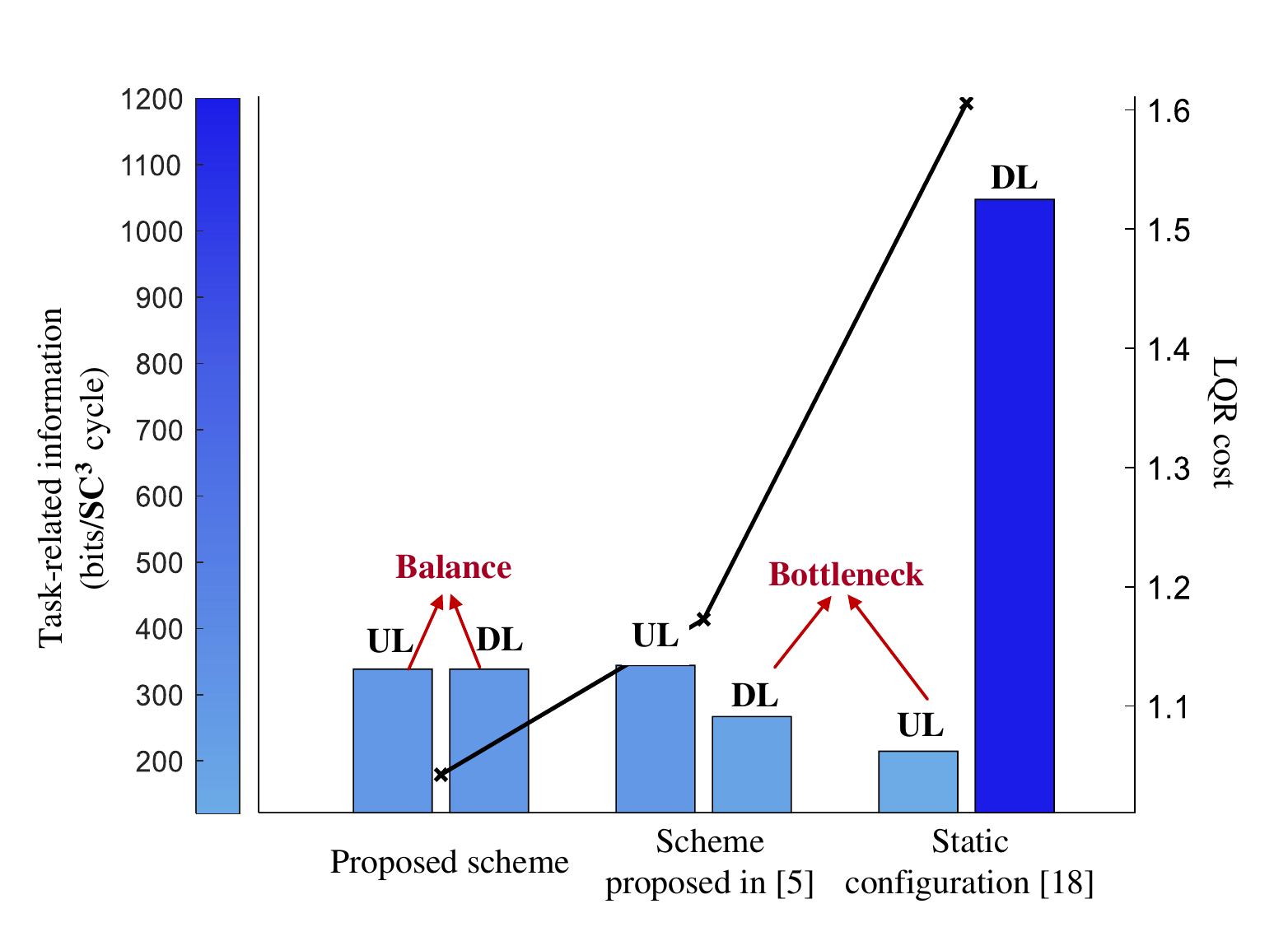}
	\caption{The UL\&DL task-related information and the LQR cost under the proposed scheme, the scheme proposed in \cite{ref1}, and the static configuration \cite{3gpp_r17}.}
	\label{simu2}
\end{figure}


Fig. \ref{simu3} depicts the LQR cost contour map under the optimal UL\&DL and computing configuration. We can see that the LQR cost decreases with the enhancement of both computing and communication resources. Furthermore, the map reveals a notable trade-off between communication bandwidth and computing CPU frequency. Take the contour line with the LQR cost of 1.2 as an example. In the bandwidth-constrained region, a 4 kHz increase in bandwidth can offset the need for 1 GHz CPU frequency. Conversely, in the CPU-constrained region, an additional 20 kHz bandwidth only compensates for 110 MHz CPU frequency. It is interesting to find the marginal utility balance of communication and computing such that the operational cost of the $\mathbf{SC}^3$ loop can be minimized.

%
%
%
%
%
\section{conclusion}

In this letter, we have investigated the basic model of the reflex-arc-like $\mathbf{SC}^3$ loop. To maintain the task-level connections between two communication links within the $\mathbf{SC}^3$ loop, we have treated the sensor and actuator as a virtual user and jointly optimized UL\&DL transmit power, time, bandwidth, and computing CPU frequency to minimize the LQR cost. We have simplified the complex problem into a convex bandwidth allocation problem, along with the optimal closed-form solution for the time allocation. Our simulation results have confirmed the superiority of the proposed task-oriented UL\&DL optimization scheme, highlighting the importance of keeping task-level UL\&DL balance for the $\mathbf{SC}^3$ loop.

\begin{figure} [t]
	\centering
	\includegraphics[width=0.9\linewidth]{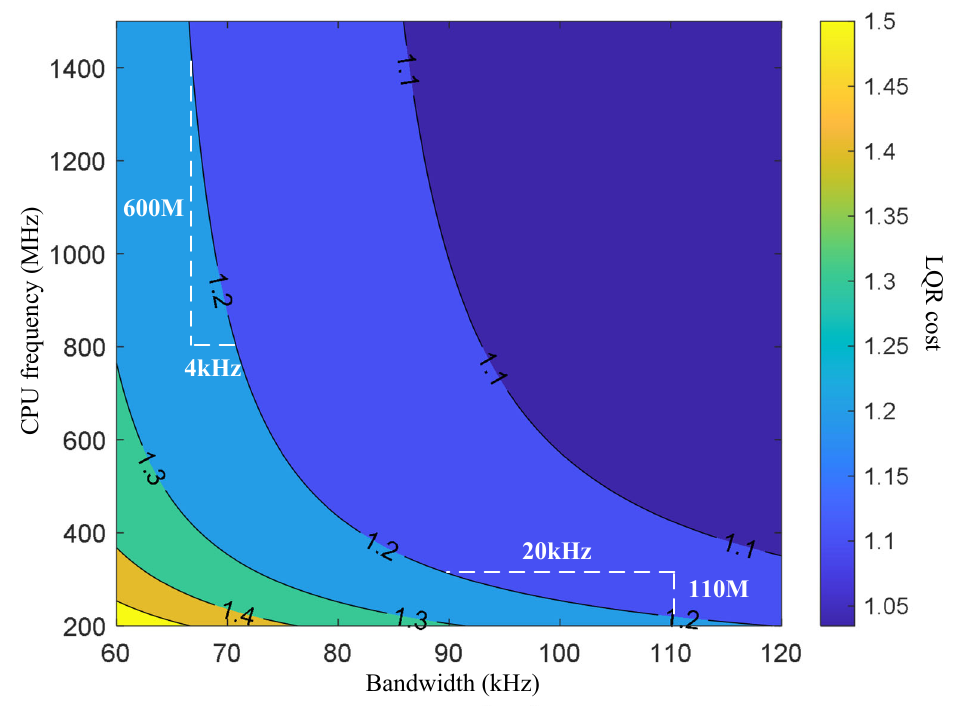}
	\caption{The LQR cost contour map varying with the bandwidth and computing CPU frequency.}
	\label{simu3}
\end{figure}



\end{document}